\documentclass[conference,letter]{IEEEtran}

\usepackage[utf8]{inputenc} 
\usepackage[T1]{fontenc}
\usepackage{url}              
\usepackage{cite}             
\usepackage{amsmath,amsfonts}
\usepackage{graphicx}
\usepackage{algorithmic}
\usepackage{array}
\usepackage[caption=false,font=normalsize,labelfont=normalsize,textfont=normalsize]{subfig}
\usepackage{amsmath,amssymb,amsfonts}
\usepackage{algorithmic}
\usepackage{graphicx}
\usepackage{multirow}
\usepackage{amsmath}
\usepackage{amssymb}
\usepackage{aligned-overset}
\usepackage{latexsym}
\usepackage{epsfig}
\usepackage{epsf}
\usepackage{graphics}
\usepackage{listings}
\usepackage{color}
\usepackage{epstopdf}
\usepackage{multirow}
\usepackage{psfrag}
\usepackage{fancyhdr}
\usepackage{multirow}
\usepackage{textcomp}
\usepackage{float}
\usepackage{subfig}
\usepackage{xcolor}
\usepackage{caption}
\usepackage{times}  
\usepackage{helvet}  
\usepackage{courier}  
\usepackage{graphicx} 
\usepackage{caption} 
\usepackage{cite}
\usepackage{algorithm}
\usepackage{algorithmic}
\usepackage{balance}
\usepackage{xurl}
\usepackage{multicol}
\usepackage{graphicx}
\usepackage{listings}
\usepackage{multicol}
\usepackage[utf8]{inputenc} 
\usepackage[T1]{fontenc}    
\usepackage{url}            
\usepackage{booktabs}       
\usepackage{amsfonts}       
\usepackage{nicefrac}       
\usepackage{microtype}      
\usepackage{xcolor}         
\usepackage{mathtools, amsmath, amsthm}
\usepackage{tikz,pgfplots}
\usepackage{newfloat}
\usepackage{listings}
\usepackage{multicol}

\newtheorem{theorem}{{Theorem}}
\newtheorem{lemma}[theorem]{{Lemma}}
\newtheorem{proposition}[theorem]{{Proposition}}

\theoremstyle{definition}
\newtheorem{definition}{Definition}
\newtheorem{remark}{{\textbf{Remark}}}
\newtheorem{example}{{Example}}
\usepackage{textcomp}
\usepackage{stfloats}
\usepackage{url}
\usepackage{verbatim}
\usepackage{graphicx}

\usepackage{graphicx}         
\usepackage{booktabs}         

\def\BibTeX{{\rm B\kern-.05em{\sc i\kern-.025em b}\kern-.08em
    T\kern-.1667em\lower.7ex\hbox{E}\kern-.125emX}}
\usepackage{balance}
\begin{document}

\IEEEoverridecommandlockouts
\title{Projective Systematic Authentication via Reed-Muller Codes
\thanks{This work was supported by the Center for Ubiquitous Connectivity (CUbiC), sponsored
by Semiconductor Research Corporation (SRC) and Defense Advanced Research Projects
Agency (DARPA) under the JUMP 2.0 program.}}

\author{%
  \IEEEauthorblockN{Hsuan-Po Liu\IEEEauthorrefmark{1} and Hessam Mahdavifar\IEEEauthorrefmark{1}\IEEEauthorrefmark{2}}
  \IEEEauthorblockA{\IEEEauthorrefmark{1}Department of Electrical Engineering and Computer Science, University of Michigan, Ann Arbor, MI 48109, USA}
  \IEEEauthorblockA{\IEEEauthorrefmark{2}Department of
Electrical and Computer Engineering, Northeastern University, Boston, MA
02115, USA}
  \IEEEauthorblockA{Emails: hsuanpo@umich.edu, h.mahdavifar@northeastern.edu}
}

\maketitle

\begin{abstract}
In this paper, we study the problem of constructing projective systematic authentication schemes based on binary linear codes. 
In systematic authentication, a tag for authentication is generated and then appended to the information, also referred to as the source, to be sent from the sender.
Existing approaches to leverage \textit{projective} constructions focus primarily on codes over large alphabets, and the projection is simply into one single symbol of the codeword.
In this work, we extend the projective construction and propose a general projection process in which the source, which is mapped to a higher dimensional codeword in a given code, is first projected to a lower dimensional vector. The resulting vector is then masked to generate the tag. To showcase the new method, we focus on leveraging binary linear codes and, in particular, Reed-Muller (RM) codes for the proposed projective construction. 
More specifically, we propose systematic authentication schemes based on RM codes, referred to as RM-A-codes. 
We provide analytical results for probabilities of deception, widely considered as the main metrics to evaluate the performance of authentication systems. 
Through our analysis, we discover and discuss explicit connections between the probabilities of deception and various properties of RM codes. 


\end{abstract}

\section{Introduction}
\label{sec:intro}

With the rapid expansion of wireless networks, the need for providing message integrity and authenticity has become increasingly crucial, and is widely regarded as one of the major goals of cryptography systems \cite{Hakeem2022SecurityRA}. Authentication codes were first introduced in \cite{gilbert1974codes}. A theoretical framework for authentication was then introduced by Simmons \cite{Simmons1985AuthenticationTT}, which considers an unconditionally secure authentication system, i.e., where the adversary may have unlimited computational power. 

A conventional authentication system involves three parties: a \emph{sender} who sends a message, a \emph{receiver} who is the intended recipient of the message, and an \emph{adversary} who attempts to attack by either impersonating the sender and inserting a message into the channel, or substituting an intercepted message with a fraudulent one. These two types of attacks are termed the \emph{impersonation} attack and the \emph{substitution} attack, respectively. The communication is assumed to take place over a public channel. To protect the system from the aforementioned attacks, the sender and the receiver utilize a shared secret key, known only to them, which is then used in the encoding rule of the underlying authentication code. The probabilities of successful impersonation and substitution attacks by the adversary are the probabilities of deception, which are considered to evaluate the performance of the authentication system.

A vast body of work is dedicated to designing authentication codes with various methodologies and under various constraints. Some of the major approaches to this problem include geometric codes \cite{bierbrauer1994families,bierbrauer1997universal}, nonlinear functions \cite{ chanson2003cartesian, ding2004systematic,carlet2006authentication}, algebraic constructions \cite{helleseth1996universal, xing2000constructions, ozbudak2006some}, and error-correcting codes \cite{ johansson1994relation, kabatianskii1996cardinality, wang2003linear, ding2005coding, ding2007generic, liu2018new}. 
In this work, we focus on authentication codes without secrecy, which are also known as \emph{systematic} authentication codes, constructed via error-correcting codes using a \emph{projective} construction \cite{ ding2005coding ,ding2007generic, liu2018new }. 
In systematic authentication codes, a message is sent from the sender to the receiver through the public channel, including the source state (i.e., plaintext), appended with a tag. The tag is generated by an encoding rule from a shared secret key between the sender and the receiver. 
The \textit{projective} constructions proposed in prior works \cite{ ding2005coding ,ding2007generic, liu2018new }, are specifically for codes over rather large underlying alphabets, and the projection is simply into one single symbol of the codeword. 

In this paper, we extend upon the projective authentication methods and consider a general notion of projecting higher dimensional codewords generated by a specific error-correcting code to a lower dimensional vector, e.g., sub-blocks of codewords. This constitutes a major building block of the system. More specifically, the secret key is split into two sub-parts. The first part is used to indicate the subset of the codeword to be projected to the lower dimension, where it is masked by the second part of the key. In order to showcase the proposed scheme, we focus on designing new projective authentication codes based on binary linear codes, and in particular, Reed-Muller (RM) codes, referred to as RM-A-codes. The main motivation behind this choice is to demonstrate that the already existing physical layer blocks for binary error correction can be leveraged for authentication as well, leading to potential solutions for low-complexity low-cost communication devices in massive networks, such as in Internet-of-Things (IoT) networks. RM codes are one of the oldest families of codes, which have received renewed attention in recent years due to their capacity-achieving properties \cite{abbe2023reed,reeves2023reed} as well as their excellent performance in short blocklengths \cite{abbe2020reed,jamali2021reed,jamali2023machine}. We demonstrate that RM codes are a perfect fit as a building block for the proposed systematic authentication based on binary linear codes, and present closed-form expressions for the probability of success of attacks by an adversary in the considered authentication system. 
As for the probability of success for substitution attacks, we show that our construction reduces the computationally expensive calculations for characterizing the guarantees, and identifies the performance under different settings considering RM codes. 

The rest of the paper is structured as follows.
In Section~\ref{sec:prelim}, we provide the preliminaries of systematic authentication codes. 
In Section~\ref{sec:protocol}, we present the proposed projective constructions.
We then analyze the theoretical results for our proposed construction in Section~\ref{sec:ana}.
Finally, we conclude the paper in Section~\ref{sec:con}.

\section{Preliminaries}
\label{sec:prelim}


\subsection{Systematic authentication codes}
The systematic authentication code is defined as a four-tuple $(\mathcal{S},\mathcal{T},\mathcal{K},\{\mathcal{E}_\mathbf{k}:\mathbf{k}\in\mathcal{K}\})$, where $\mathcal{S}$ is the source space, $\mathcal{T}$ is the tag space, $\mathcal{K}$ is the key space, and $\mathcal{E}_\mathbf{k}:\mathcal{S}\mapsto\mathcal{T}$ is the encoding rule.
During the authentication phase, which may happen only once at the beginning of the communication or at the beginning of every new round of communication, the sender first generates a source $\mathbf{s}\in\mathcal{S}$. Then given the secretly shared key $\mathbf{k}$ between the sender and the receiver, the encoding rule $\mathcal{E}_\mathbf{k}:\mathcal{S}\mapsto\mathcal{T}$ generates a tag $\mathbf{t}=\mathcal{E}_{\mathbf{k}}(\mathbf{s})\in\mathcal{T}$. The message $\mathbf{m}\in\mathcal{M}=\mathcal{S\times T}$ sent from the sender to the receiver is then denoted by the concatenation of $\mathbf{s}$ and $\mathbf{t}$ as $\mathbf{m}=[\mathbf{s},\mathbf{t}]$. 
When the receiver receives a message (which includes a source vector and a tag vector), it checks the authenticity by verifying whether the received tag can be generated from the secretly shared key $\mathbf{k}$ through the encoding rule $\mathcal{E}_\mathbf{k}:\mathcal{S}\mapsto\mathcal{T}$ or not. If yes, the receiver accepts the received message; otherwise, the receiver discards it.
We summarize the systematic authentication system in Fig.~\ref{fig:sys}.
Note that we will specify the dimension for each of the underlying vectors in the next section when we formally propose our construction.

\begin{figure*}[t]
\centering
{\includegraphics[width=.93\textwidth]{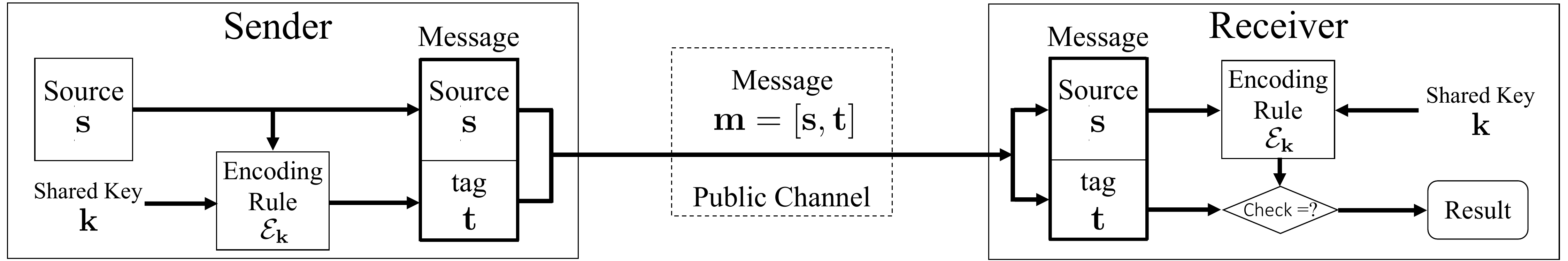}}\\
  \caption{Systematic authentication system}\label{fig:sys}
\end{figure*}

The adversary is assumed to have the ability to insert messages into the public channel and/or to intercept messages that are sent over the public channel and to modify them.
Two types of attacks are often considered in the authentication systems, referred to as the impersonation attack and the substitution attack. An impersonation attack occurs when the adversary inserts a new message $\mathbf{m}^\prime=[\mathbf{s}^\prime,\mathbf{t}^\prime]$ into the public channel, see Fig.~\ref{fig:PI}.
A substitution attack is when the adversary observes a message $\mathbf{m}=[\mathbf{s},\mathbf{t}]$ that exists in the public channel, intercepts it, then inserts a new message $\mathbf{m}^\prime=[\mathbf{s}^\prime,\mathbf{t}^\prime]$ into the channel, where $\mathbf{s}^\prime\neq\mathbf{s}$, this attack is demonstrated in Fig.~\ref{fig:PS}. 
The security guarantees of the system are measured in terms of the adversary’s probability of success with respect to each of the attacks. The probabilities of success for the impersonation and the substitution attacks by the adversary are denoted by $P_\mathrm{I}$ and $P_\mathrm{S}$, respectively. These quantities are defined more explicitly in the next subsection. 


\begin{figure}
\captionsetup{justification=centering}
\centering
\subfloat[Impersonation attack]{%
\includegraphics[width=0.24\textwidth]{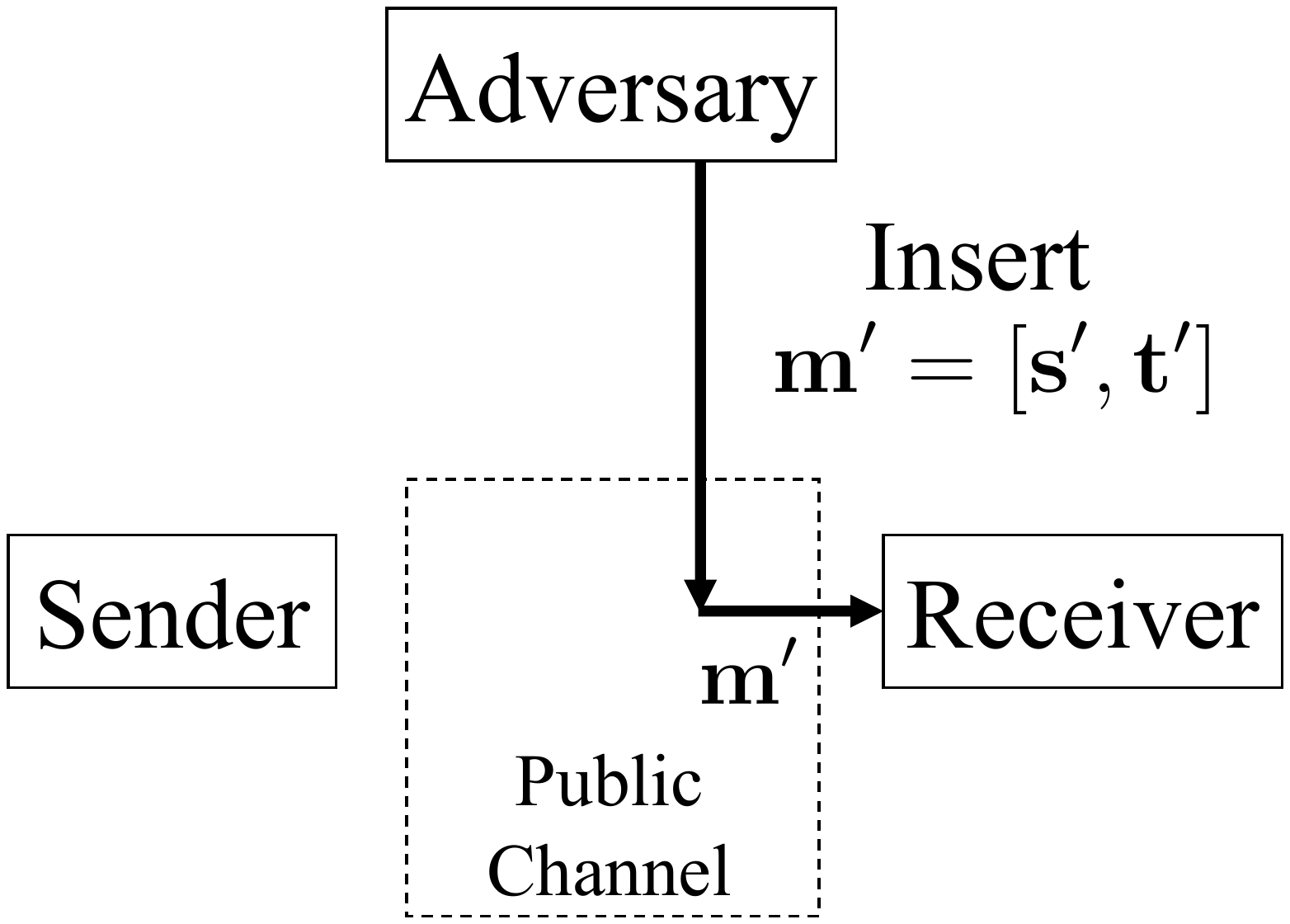}\label{fig:PI}}\hfill
\subfloat[Substitution attack]{
\includegraphics[width=0.24\textwidth]{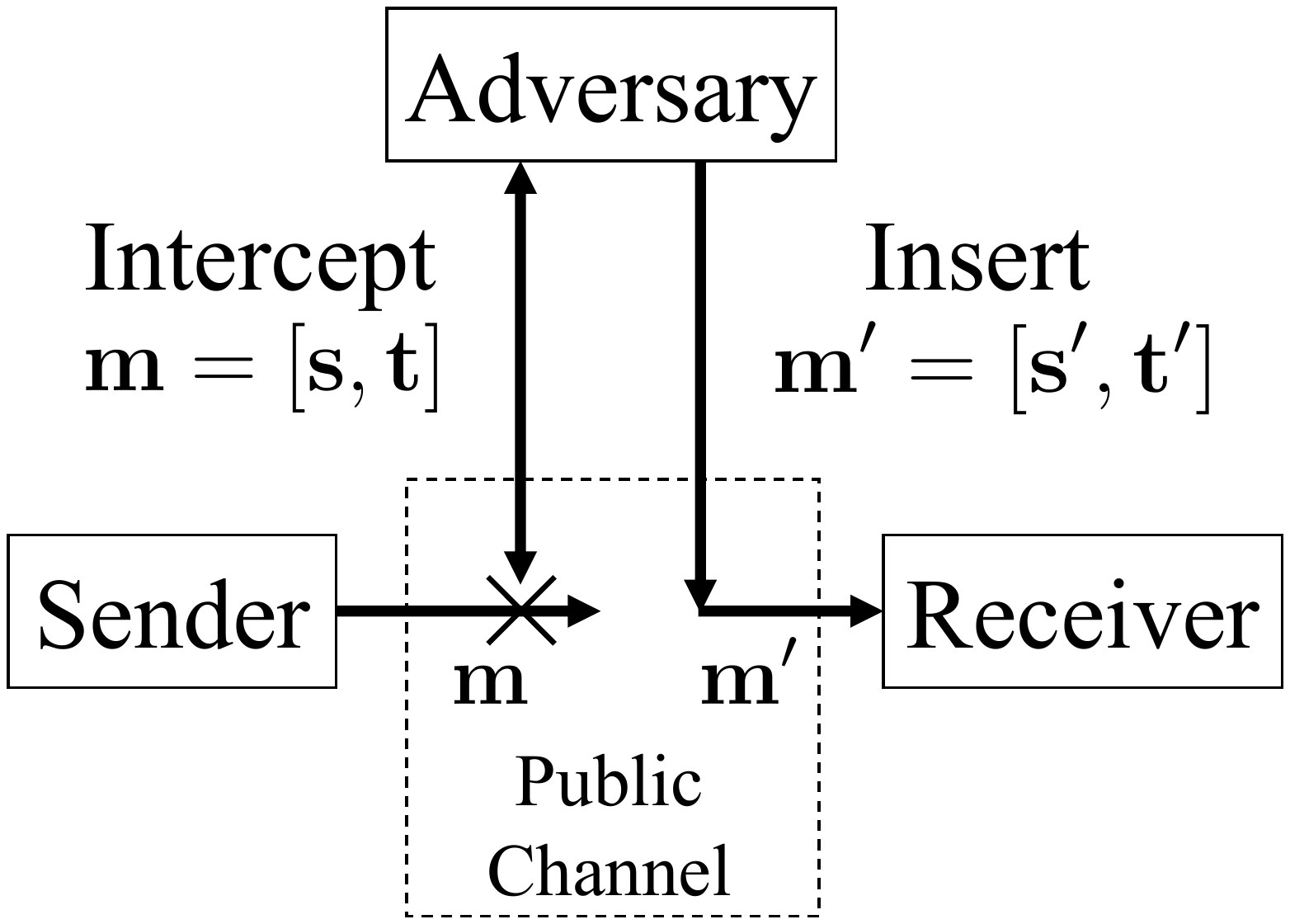}\label{fig:PS}}\hfill
\caption{Two types of attacks}\label{fig:attacks}
\end{figure}

\subsection{Probability of deception}
The probability of success of the impersonation attack is denoted by $P_\mathrm{I}$ and defined as

\begin{equation}
\begin{aligned}
P_\mathrm{I}:=&\max_{\mathbf{s}^\prime,\mathbf{t}^\prime}\,\mathbb{P}\left([\mathbf{s}^\prime,\mathbf{t}^\prime]\textrm{\,valid}\right),
\end{aligned}
\end{equation}
where
$$\mathbb{P}([\mathbf{s}^\prime,\mathbf{t}^\prime] \textrm{\,valid}) = \frac{|\{\mathbf{k}\in\mathcal{K}:\mathbf{t}^\prime=\mathcal{E}_\mathbf{k}(\mathbf{s}^\prime)\}|}{|\{\mathbf{k}\in\mathcal{K}\}|}.$$
Note that this probability is with respect to the space of all possible keys under a uniform distribution. 


The probability of success of the substitution attack is denoted by $P_\mathrm{S}$ and defined as

\begin{equation}
\label{eq:P_S_worst}
\begin{aligned}
P_\mathrm{S}:=\max_{\mathbf{s},\mathbf{t}}\max_{\mathbf{s}^\prime\neq\mathbf{s},\mathbf{t}^\prime}\,\mathbb{P}([\mathbf{s}^\prime,\mathbf{t}^\prime]  \textrm{\,valid}\,|\,[\mathbf{s},\mathbf{t}]  \textrm{\,observed}),
\end{aligned}
\end{equation}
where 
\begin{equation}
\label{eq:P_S_payoff}
    \mathbb{P}([\mathbf{s}^\prime,\mathbf{t}^\prime] \textrm{\,valid}\,|\,[\mathbf{s},\mathbf{t}] \textrm{\,observed})=
\frac{\left|\left\{\mathbf{k}\in\mathcal{K}:\begin{matrix*}[l]
    \mathbf{t}=\mathcal{E}_\mathbf{k}(\mathbf{s}),
    \\
    \mathbf{t}^\prime=\mathcal{E}_\mathbf{k}(\mathbf{s}^\prime)
\end{matrix*}\right\}\right|}
{|\{\mathbf{k}\in\mathcal{K}:\mathbf{t}=\mathcal{E}_\mathbf{k}(\mathbf{s})\}|}.
\end{equation}
In this analysis, it is assumed that both the key and the source state are from uniform distributions on the key space and the source space, respectively.

It is shown in \cite{stinson1992combinatorial} that we have $P_\mathrm{S}\geq P_\mathrm{I}\geq \frac{1}{|\mathcal{T}|}$. The core design criterion for constructing a good systematic authentication code is minimizing the probabilities of deception $P_\mathrm{I}$ and $ P_\mathrm{S}$.

\subsection{RM codes: A brief overview}
An RM code \cite{abbe2020reed} is denoted by $\mathrm{RM}(m,r)$, where $m$ is a positive integer that defines the blocklength of the code as $n=2^m$, $r$ is termed as the order of the code with $r\in\{0,1,\dots,m\}$, which determines the dimension of the code as $\sum_{i=0}^r{m \choose i}$. 
Note that $\mathrm{RM}(m,r)$ generates a $(2^m,\sum_{i=0}^r{m \choose i})$-code. The resulting generator matrix $\mathbf{G}$ is constructed by 
\begin{equation}
\label{eq:gmatrix}
    \mathbf{G}=
    \begin{bmatrix}
    \mathbf{G}_r\\
    \vdots\\ 
    \mathbf{G}_{1}\\ 
    \mathbf{G}_0
    \end{bmatrix}_{\sum_{i=0}^r{m \choose i}\times 2^m},
\end{equation}
where $\mathbf{G}_0$ is an all-one vector with $n$ entries, $\mathbf{G}_1$ is a matrix of dimension $m\times n$ that each column vector is a unique $m$-bit binary vector, and $\mathbf{G}_j$ is an ${m \choose j}\times n$ matrix that each row is constructed by an entry-wise product of a distinct set of $j$ rows from $\mathbf{G}_1$, for $j\in\{2,3,\dots,r\}$. Thus, $\mathbf{G}$ has $\sum_{i=0}^r{m \choose i}$ rows, with a minimum distance of $2^{m-r}$. 

Note also that the last row of each of the  $\mathbf{G}_i$'s, for $i\in\{1,2,\dots,r\}$, are in the form of $[\mathbf{1}_{2^{m-i}},\mathbf{0}_{n-2^{m-i}}]$. Such vectors, consisting of two separate sub-blocks of all-ones and all-zeros, play a critical role in the analysis of the probability of success for the substitution attacks. 
This is an advantage of employing RM-A-codes, where the set of such codewords is exactly known and can be characterized.

\section{The Proposed RM-A-codes}
\label{sec:protocol}
First, recall a systematic authentication code, defined in Section~\ref{sec:prelim}-A, as a four-tuple $(\mathcal{S},\mathcal{T},\mathcal{K},\{\mathcal{E}_\mathbf{k}:\mathbf{k}\in\mathcal{K}\})$. We define the proposed projective construction based on RM codes, referred to as RM-A-codes, as follows:

\begin{definition}[RM-A-codes]
Consider the source $\mathbf{s}\in\mathcal{S}=\{0,1\}^M$, tag $\mathbf{t}\in\mathcal{T}=\{0,1\}^l$, and let the
key $\mathbf{k}\in\mathcal{K}$ to be a concatenation of two keys $\mathbf{k}_1\in\mathcal{K}_1\subset\{0,1\}^n$ and $\mathbf{k}_2\in\mathcal{K}_2=\{0,1\}^l$, such that $\mathbf{k}=[\mathbf{k}_1,\mathbf{k}_2]\in\mathcal{K}=\mathcal{K}_1\times\mathcal{K}_2$, for some positive integers $n,M,l$ with $n> M \geq l$. More details on $\mathcal{K}_1$ are discussed later in Remark~\ref{rmk:k_1}.
The tag $\mathbf{t}$ is defined as $\mathbf{t}=\mathcal{E}_{\mathbf{k}}(\mathbf{s})=\mathbf{c}_{\mathbf{s},\mathbf{k}_1}+\mathbf{k}_2$, where $\mathbf{c}_{\mathbf{s},\mathbf{k}_1}$ is a vector selected as a subset of $l$ entries from the entries of the codeword $\mathbf{c}_{\mathbf{s}}$ with the indices determined by $\mathbf{k}_1$, $\mathbf{c}_{\mathbf{s}}$ is a codeword encoded from the source $\mathbf{s}\in\{0,1\}^M$ as a subset of the information input $\mathbf{u}\in\{0,1\}^{\scriptscriptstyle \sum_{i=0}^r{m \choose i}}$, such that $\mathbf{u}=[\mathbf{0}_{\scriptscriptstyle \sum_{i=0}^r{m \choose i}-M-1},\mathbf{s},0]$, for a $(2^m,\sum_{i=0}^r{m \choose i})$-code from $\mathrm{RM}(m,r)$ with a generator matrix $\mathbf{G}$, where $M<\sum_{i=0}^r{m \choose i}$, such that $\mathbf{c_s=uG}$.
\end{definition}

\begin{remark}
\label{rmk:k_1}
    The key $\mathbf{k}_1$ is a length-$n$ binary vector with a weight of $l$, i.e., containing $l$ ones, where the indices of ones indicate the indices of the selected $l$ entries from $\mathbf{c}_{\mathbf{s}}$ to construct $\mathbf{c}_{\mathbf{s},\mathbf{k}_1}$. 
    Thus, given a systematic authentication code, we have $|\mathcal{K}_1|={n \choose l}$. Note that we may reduce the length of $\mathbf{k}_1$ to $\lceil\log_2{|\mathcal{K}_1|}\rceil=\lceil\log_2{{n \choose l}}\rceil$.
\end{remark}

To generate $\mathbf{c}_{\mathbf{s}}$ given $\mathbf{s}$ and $\mathbf{G}$ from $\mathrm{RM}(m,r)$ in \eqref{eq:gmatrix}, one needs to specify $\mathbf{u}$. 
To this end, we start by presenting the following lemmas.
\begin{lemma}
\label{lemma:u_complement}
Given a generator matrix $\mathbf{G}$ from $\mathrm{RM}(m,r)$, we have $\mathbf{c}_\mathbf{s}+\mathbf{c}_{\mathbf{s}^\prime}=\mathbf{1}_n$ when $\mathbf{u=s}$ and $\mathbf{u}^\prime=\mathbf{s}^\prime$, where $\mathbf{c}_\mathbf{s}=\mathbf{uG}$ and $\mathbf{c}_{\mathbf{s}^\prime}=\mathbf{u}^\prime\mathbf{G}$, if $\mathbf{s}$ and $\mathbf{s}^\prime$ are binary vectors that differ only in the last entry (referred to as neighboring vectors).
\end{lemma}
\begin{proof}
Let $\mathbf{u}=[u_1,u_2,\dots,u_{\scriptscriptstyle \sum_{i=0}^r{m \choose i}-1},0]$ and $\mathbf{u}^\prime=[u_1,u_2,\dots,u_{\scriptscriptstyle \sum_{i=0}^r{m \choose i}-1},1]$ to be a pair of binary neighboring vectors. Then we have
    \begin{equation}
    \label{eq:complement_u}
        \mathbf{c}_\mathbf{s}+\mathbf{c}_{\mathbf{s}^\prime}=\mathbf{uG}+\mathbf{u}^\prime\mathbf{G}=(\mathbf{u}+\mathbf{u}^\prime)\mathbf{G}=[\mathbf{0}_{\scriptscriptstyle \sum_{i=0}^r{m \choose i}-1},1]\mathbf{G}.
    \end{equation}
Equivalently, $\mathbf{c}_\mathbf{s}+\mathbf{c}_{\mathbf{s}^\prime}$ is equal to the last row in $\mathbf{G}$, which is the all-one vector according to \eqref{eq:gmatrix}. 
\end{proof}

The next lemma demonstrates that using plain RM codes results in the probability of success for the substitution attack being one. Hence, we will modify the structure by considering sub-codes of RM codes, i.e., by letting $\mathbf{s}$ be a sub-vector of $\mathbf{u}$ while the remaining entries of $\mathbf{u}$ are fixed to zeros. This will be clarified later. 

\begin{lemma}
\label{lemma:Ps_1}
    For any codeword $\mathbf{c}_\mathbf{s}$, the adversary can pick $\mathbf{c}_{\mathbf{s}^\prime}$, where $\mathbf{s},\mathbf{s}^\prime\in\mathcal{S}$, with $\mathbf{c}_\mathbf{s}+\mathbf{c}_{\mathbf{s}^\prime}=\mathbf{1}_n$, resulting in  $P_\mathrm{S}=1$.
\end{lemma}
\begin{proof}
    Given $\mathbf{k}_1$, when $\mathbf{c}_\mathbf{s}+\mathbf{c}_{\mathbf{s}^\prime}=\mathbf{1}_n$, we have
    \begin{equation}
    \label{eq:thm_sum_1}\mathbf{c}_{\mathbf{s},\mathbf{k}_1}+\mathbf{c}_{\mathbf{s}^\prime,\mathbf{k}_1}=\mathbf{1}_l.
    \end{equation}
    Note that, 
    \begin{equation}
        \label{eq:thm_sum_2}\mathbf{c}_{\mathbf{s},\mathbf{k}_1}+\mathbf{c}_{\mathbf{s}^\prime,\mathbf{k}_1}=
        (\mathbf{c}_{\mathbf{s},\mathbf{k}_1}+\mathbf{k}_2)+(\mathbf{c}_{\mathbf{s}^\prime,\mathbf{k}_1}+\mathbf{k}_2)=\mathbf{t}+\mathbf{t}^\prime.
    \end{equation}
    Combining \eqref{eq:thm_sum_1} and \eqref{eq:thm_sum_2}, we obtain that $\mathbf{t}^\prime=\mathbf{t}+\mathbf{1}_l$, when $\mathbf{c}_{\mathbf{s},\mathbf{k}_1}+\mathbf{c}_{\mathbf{s}^\prime,\mathbf{k}_1}=\mathbf{1}_l$. In order to compute $P_\mathrm{S}$ defined in \eqref{eq:P_S_worst}, first we derive that
    \begin{equation}
    \begin{aligned}
        \mathbb{P}([\mathbf{s}^\prime&,\mathbf{t}^\prime] \textrm{\,valid}\,|\,[\mathbf{s},\mathbf{t}] \textrm{\,observed})\\&=
\frac{\left|\left\{\mathbf{k}\in\mathcal{K}:
    \mathbf{t}=\mathcal{E}_\mathbf{k}(\mathbf{s}),
    \mathbf{t}^\prime=\mathcal{E}_\mathbf{k}(\mathbf{s}^\prime)\right\}\right|}
{|\{\mathbf{k}\in\mathcal{K}:\mathbf{t}=\mathcal{E}_\mathbf{k}(\mathbf{s})\}|}\\
&=\frac{\left|\left\{\mathbf{k}\in\mathcal{K}:
    \mathbf{t}=\mathcal{E}_\mathbf{k}(\mathbf{s}),
    \mathbf{t}^\prime+
    \mathbf{1}_l=\mathcal{E}_\mathbf{k}(\mathbf{s}^\prime)+
    \mathbf{1}_l\right\}\right|}
{|\{\mathbf{k}\in\mathcal{K}:\mathbf{t}=\mathcal{E}_\mathbf{k}(\mathbf{s})\}|}\\
\overset{\scriptstyle\mathrm{(a)}}&{=}\frac{\left|\left\{\mathbf{k}\in\mathcal{K}:
    \mathbf{t}=\mathcal{E}_\mathbf{k}(\mathbf{s}),
    \mathbf{t}=\mathcal{E}_\mathbf{k}(\mathbf{s})\right\}\right|}
{|\{\mathbf{k}\in\mathcal{K}:\mathbf{t}=\mathcal{E}_\mathbf{k}(\mathbf{s})\}|}\\
&=\frac{\left|\left\{\mathbf{k}\in\mathcal{K}:
    \mathbf{t}=\mathcal{E}_\mathbf{k}(\mathbf{s})\right\}\right|}
{|\{\mathbf{k}\in\mathcal{K}:\mathbf{t}=\mathcal{E}_\mathbf{k}(\mathbf{s})\}|}=1,
    \end{aligned}
    \end{equation}
    where $\mathrm{(a)}$ is according to $\mathbf{t}^\prime=\mathbf{t}+\mathbf{1}_l$, when $\mathbf{c}_{\mathbf{s},\mathbf{k}_1}+\mathbf{c}_{\mathbf{s}^\prime,\mathbf{k}_1}=\mathbf{1}_l$.
    Therefore, with $\mathbb{P}([\mathbf{s}^\prime,\mathbf{t}^\prime] \textrm{\,valid}\,|\,[\mathbf{s},\mathbf{t}]\textrm{\,observed})=1$, we have $P_\mathrm{S}=1$.
\end{proof}

With Lemmas~\ref{lemma:u_complement} and~\ref{lemma:Ps_1}, we have the following proposition.

\begin{proposition}
\label{prop:u}
    RM-A-codes ensure $P_\mathrm{S}<1$ if we set the last entry in $\mathbf{u}$ frozen to be $0$.
\end{proposition}
\begin{proof}
    Lemma~\ref{lemma:u_complement} indicates that we will have $\mathbf{c}_\mathbf{s}+\mathbf{c}_{\mathbf{s}^\prime}=\mathbf{1}_n$ if we do not set the last entry of $\mathbf{u}$ to be $0$. Consequently, Lemma~\ref{lemma:Ps_1} implies that it will lead to an authentication code with $P_\mathrm{S}=1$, which is undesirable. Except $\mathbf{c}_{\mathbf{s}^\prime}$, there is no other codeword (together with a tag), that form a valid pair for all choices of the key. Hence, removing $\mathbf{c}_{\mathbf{s}^\prime}$ results in $P_\mathrm{S}<1$.  
    
\end{proof}

Proposition~\ref{prop:u} implies that we should set the last entry of $\mathbf{u}$ frozen to be $0$, since the core design criterion for constructing a good authentication code is to minimize the probabilities of deception and one should naturally avoid $P_\mathrm{S}=1$. 

As specified by Proposition~\ref{prop:u} with $M<\sum_{i=0}^r{m \choose i}$, we obtain $\mathbf{u}=[\mathbf{0}_{\scriptscriptstyle \sum_{i=0}^r{m \choose i}-M-1},\mathbf{s},0]$
such that $\mathbf{c}_\mathbf{s}=\mathbf{uG}=[\mathbf{0}_{\scriptscriptstyle \sum_{i=0}^r{m \choose i}-M-1},\mathbf{s},0]\mathbf{G}$.
To construct an RM-A-code as defined in Definition 1, we need to determine 
the length of the source $M$, the length of the tag $l$, and $\mathrm{RM}(m,r)$ to generate a $(2^m,\sum_{i=0}^r{m \choose i})$-code which has a blocklength $n=2^m$. 
We end this section by illustrating a toy example of RM-A-codes.

\begin{example}
\label{exp:toy_construct}
    Let $M=2$, $l=1$, and choose $m=2,r=1$ to generate a $(4,3)$-code by $\mathrm{RM}(2,1)$. In such code, we obtain the generator matrix as
\begin{equation}
    \mathbf{G}=\begin{bmatrix}
        1 & 0 & 1 & 0\\
        1 & 1 & 0 & 0\\
        1 & 1 & 1 & 1
    \end{bmatrix}.
\end{equation}
The information input $\mathbf{u}$ is a length-$3$ vector as $\mathbf{u}=[\mathbf{s},0]=[s_1,s_2,0]$,
where $\mathbf{s}=[s_1,s_2]\in\mathcal{S}=\{[0,0],[0,1],[1,0],[1,1]\}$.
We obtain the codeword as $\mathbf{c}_\mathbf{s}
=\mathbf{uG}=[s_1+s_2,s_2,s_1,0]$.
The keys $\mathbf{k}_1$ and $\mathbf{k}_2$ are constructed as $\mathbf{k}_1\in\mathcal{K}_1=\{[1,0,0,0],[0,1,0,0],[0,0,1,0],[0,0,0,1]\}$ and $\mathbf{k}_2\in\mathcal{K}_2=\{[0],[1]\}$.
Then, the tags can be generated from the given key, recall that $\mathbf{t}=\mathbf{c}_{\mathbf{s},\mathbf{k}_1}+\mathbf{k}_2$. Then, for instance, given $\mathbf{k}_1=[1,0,0,0]$ and $\mathbf{k}_2=[0]$, we have $\mathbf{t}=\mathbf{c}_{\mathbf{s},[1,0,0,0]}+[0]=[s_1+s_2]$,
\textcolor{black}{where $\mathbf{c}_{\mathbf{s},[1,0,0,0]}$ denotes selecting the first entry in the codeword $\mathbf{c}_{\mathbf{s}}$.} By calculating the tags $\mathbf{t}$'s over all sources in $\mathcal{S}$ and all keys $\mathcal{K}=\mathcal{K}_1\times\mathcal{K}_2$, we have the authentication matrix shown in TABLE~\ref{tab:eg_A-matrix}. 
Applying the definitions of the probabilities of deception yields $ P_\mathrm{I}=0.5$, $P_\mathrm{S}=0.5$.

\begin{table}[t]  
\normalsize
\centering
\begin{tabular}{|c || c | c | c | c|}  
 \hline
 & $\mathbf{c}_{[0,0]}$ & $\mathbf{c}_{[0,1]}$ & $\mathbf{c}_{[1,0]}$ & $\mathbf{c}_{[1,1]}$ \\
 \hline\hline
 $\mathbf{c}_{\mathbf{s},[1,0,0,0]}+[0]$ & $0$ & $1$ & $1$ & $0$ \\ 
 \hline
 $\mathbf{c}_{\mathbf{s},[1,0,0,0]}+[1]$ & $1$ & $0$ & $0$ & $1$ \\
 \hline
 $\mathbf{c}_{\mathbf{s},[0,1,0,0]}+[0]$ & $0$ & $0$ & $1$ & $1$ \\
 \hline
 $\mathbf{c}_{\mathbf{s},[0,1,0,0]}+[1]$ & $1$ & $1$ & $0$ & $0$ \\
 \hline
 $\mathbf{c}_{\mathbf{s},[0,0,1,0]}+[0]$ & $0$ & $1$ & $0$ & $1$ \\ 
 \hline
 $\mathbf{c}_{\mathbf{s},[0,0,1,0]}+[1]$ & $1$ & $0$ & $1$ & $0$ \\
 \hline
 $\mathbf{c}_{\mathbf{s},[0,0,0,1]}+[0]$ & $0$ & $0$ & $0$ & $0$ \\
 \hline
 $\mathbf{c}_{\mathbf{s},[0,0,0,1]}+[1]$ & $1$ & $1$ & $1$ & $1$ \\ 
  \hline
\end{tabular}
\caption{Authentication matrix for the toy example}
  \label{tab:eg_A-matrix}
\end{table}
\end{example}

This example shows a simple case for our construction. 
In the next section, we analyze theoretical closed-form expressions for $P_\mathrm{I}$ and  $P_{\mathrm{S}}$.

\section{Analysis}
\label{sec:ana}

\subsection{The probability of success of the impersonation attack $P_\mathrm{I}$}

The following theorem shows that the proposed scheme has the lowest possible probability of success of the impersonation attack, i.e., the adversary cannot do better than a random assignment of the tag. 

\begin{theorem}
\label{thm:pi}
Given an RM-A-code as defined in Definition 1, we have
    \begin{equation}
    P_\mathrm{I}=\max_{\mathbf{s}^\prime,\mathbf{t}^\prime}\,\frac{|\{\mathbf{k}\in\mathcal{K}:\mathbf{t}^\prime=\mathcal{E}_\mathbf{k}(\mathbf{s}^\prime)=\mathbf{c}_{\mathbf{s}^\prime,\mathbf{k}_1}+\mathbf{k}_2\}|}{|\{\mathbf{k}\in\mathcal{K}\}|}=\frac{1}{2^l}.
    \end{equation}
\end{theorem}
\begin{proof}
    Since $|\mathcal{K}_1|={n \choose l}$ and $|\mathcal{K}_2|=2^l$, we have 
    \begin{equation}
    \label{eq:P_I_Nom}
        |\{\mathbf{k}\in\mathcal{K}\}|=|\mathcal{K}|=|\mathcal{K}_1\times\mathcal{K}_2|={n \choose l}\cdot 2^l.
    \end{equation}
    Furthermore, with $\mathcal{T=K}_2$, we have 
    \begin{equation}
        \label{eq:P_I_Denom}
        |\{\mathbf{k}\in\mathcal{K}:\mathbf{t}^\prime=\mathcal{E}_\mathbf{k}(\mathbf{s}^\prime)=\mathbf{c}_{\mathbf{s}^\prime,\mathbf{k}_1}+\mathbf{k}_2\}|=|\mathcal{K}_1|={n \choose l}.
    \end{equation}
    Thus, with \eqref{eq:P_I_Nom} and \eqref{eq:P_I_Denom}, we have $P_\mathrm{I}=\frac{1}{2^l}$.
\end{proof}

\subsection{Analysis for $P_{\mathrm{S}}$: The probability of success of the substitution attack}





The following theorem presents a simplified form for the calculation of $P_{\mathrm{S}}$, enabling a more efficient method for calculating and characterizing the probability of success of the substitution attack, with the assistance of a straightforward linearity property of linear codes which turns out to be useful in the analysis of the proposed RM-A-codes in Appendix~\ref{app:linearity}. 

\begin{theorem}
\label{thm:P_S_ind}
The quantity $P_{\mathrm{S}}$, defined in \eqref{eq:P_S_worst}, can be equivalently computed as 
\begin{equation}
\label{eq:P_S_w_thm_simple}
\begin{aligned}
P_{\mathrm{S}}=\max_{\Tilde{\mathbf{s}}\neq \mathbf{0}_M}\max_{\Tilde{\mathbf{t}}} \frac{|\{{\mathbf{k}_1\in\mathcal{K}_1:\mathbf{c}_{\Tilde{\mathbf{s}},\mathbf{k}_1}=\Tilde{\mathbf{t}}\}|}}{{n \choose l}},
\end{aligned}
\end{equation}
where $\Tilde{\mathbf{s}}\in\mathcal{S}$, excluding the all-zero source vector, and $\Tilde{\mathbf{t}}\in\mathcal{T}$.
\end{theorem}

\begin{proof}
We have
\begin{equation}
\label{eq:P_S_w_simplify_proof}
\begin{aligned}
P_{\mathrm{S}}
&=\max_{\mathbf{s},\mathbf{t}}\max_{\mathbf{s}^\prime\neq\mathbf{s},\mathbf{t}^\prime}\mathbb{P}([\mathbf{s}^\prime,\mathbf{t}^\prime] \textrm{\,valid}\,|\,[\mathbf{s},\mathbf{t}] \textrm{\,observed})\\
&=\max_{\mathbf{s},\mathbf{t}}\max_{\mathbf{s}^\prime\neq\mathbf{s},\mathbf{t}^\prime}\frac{\left|\left\{\mathbf{k}\in\mathcal{K}:\begin{matrix*}[l]
    \mathbf{t}=\mathbf{c}_{\mathbf{s,k}_1}+\mathbf{k}_2,
    \\
    \mathbf{t}^\prime=\mathbf{c}_{\mathbf{s}^\prime,\mathbf{k}_1}+\mathbf{k}_2
\end{matrix*}\right\}\right|}
{|\{\mathbf{k}\in\mathcal{K}:\mathbf{t}=\mathbf{c}_{\mathbf{s,k}_1}+\mathbf{k}_2\}|}\\
\overset{\scriptstyle\mathrm{(b)}}&{=}\max_{\mathbf{s},\mathbf{t}}\max_{\mathbf{s}^\prime\neq\mathbf{s},\mathbf{t}^\prime}\frac{\left|\left\{\mathbf{k}\in\mathcal{K}:\begin{matrix*}[l]
    \mathbf{t}=\mathbf{c}_{\mathbf{s,k}_1}+\mathbf{k}_2,
    \\
    \mathbf{t}+\mathbf{t}^\prime=\mathbf{c}_{\mathbf{s}+\mathbf{s}^\prime,\mathbf{k}_1}
\end{matrix*}\right\}\right|}
{{n \choose l}}\\
&=\max_{\mathbf{s},\mathbf{t}}\max_{\mathbf{s}^\prime\neq\mathbf{s},\mathbf{t}^\prime}\frac{\left|\left\{\mathbf{k}_1\in\mathcal{K}_1:\mathbf{t}+\mathbf{t}^\prime=\mathbf{c}_{\mathbf{s}+\mathbf{s}^\prime,\mathbf{k}_1}\right\}\right|}
{{n \choose l}}\\
\overset{{\scriptstyle\mathrm{(c)}}}&{=}\max_{\Tilde{\mathbf{s}}\neq\mathbf{0}_M}\max_{\Tilde{\mathbf{t}}}
\frac{\left|\left\{\mathbf{k}_1\in\mathcal{K}_1: \Tilde{\mathbf{t}}=\mathbf{c}_{\Tilde{\mathbf{s}},\mathbf{k}_1}\right\}\right|}
{{n \choose l}},
\end{aligned}
\end{equation}
where $\mathrm{(b)}$ is by the linearity of the codes that $\mathbf{c}_{\mathbf{s},\mathbf{k}_1}+\mathbf{c}_{\mathbf{s}^\prime,\mathbf{k}_1}=\mathbf{c}_{\mathbf{s}+\mathbf{s}^\prime,\mathbf{k}_1}$, together with noting that  ${|\{\mathbf{k}\in\mathcal{K}:\mathbf{t}=\mathbf{c}_{\mathbf{s,k}_1}+\mathbf{k}_2\}|}={n\choose l}$, which is equivalent to \eqref{eq:P_I_Denom}; 
in $\mathrm{(c)}$ we let $\Tilde{\mathbf{s}}=\mathbf{s}+\mathbf{s}^\prime$ and $\Tilde{\mathbf{t}}=\mathbf{t}+\mathbf{t}^\prime$. 
Note that since $\mathbf{s}\neq\mathbf{s}^\prime$, we must have $\Tilde{\mathbf{s}}\neq\mathbf{0}_M$.
\end{proof}

\begin{remark}
Theorem~\ref{thm:P_S_ind} implies that, in RM-A-codes, the expression for $P_{\mathrm{S}}$ in \eqref{eq:P_S_worst} involving two maximizations over all codewords can be simplified to \eqref{eq:P_S_w_thm_simple} which involves only a maximization over all nonzero codewords with calculating the maximum number of appearance of the valid tags corresponding to each codeword. Later, we show this can be even more simplified to a search only over the values of the weight of the codewords.

\end{remark}
For further analysis, we define the maximum probability of appearance of a valid tag $\Tilde{\mathbf{t}}$ in a nonzero codeword $\mathbf{c}_{\Tilde{\mathbf{s}}}$ given the key $\mathbf{k}_1$ as 
\begin{equation}
    P_\mathrm{t}(\mathbf{c}_{\Tilde{\mathbf{s}}})=\max_{\Tilde{\mathbf{t}}}
\frac{\left|\left\{\mathbf{k}_1\in\mathcal{K}_1: \Tilde{\mathbf{t}}=\mathbf{c}_{\Tilde{\mathbf{s}},\mathbf{k}_1}\right\}\right|}
{{n \choose l}},
\end{equation}
where $\left|\left\{\mathbf{k}_1\in\mathcal{K}_1: \Tilde{\mathbf{t}}=\mathbf{c}_{\Tilde{\mathbf{s}},\mathbf{k}_1}\right\}\right|$ can be regarded as the number of appearance of $\Tilde{\mathbf{t}}$ in the coordinates of $\mathbf{c}_{\Tilde{\mathbf{s}}}$, such that, according to Theorem~\ref{thm:P_S_ind}, we have
$$P_{\mathrm{S}}=\underset{\mathbf{c}_{\Tilde{\mathbf{s}}}}{\max}\,P_{\mathrm{t}}(\mathbf{c}_{\Tilde{\mathbf{s}}}).$$
For instance, considering Example~\ref{exp:toy_construct} in the previous section, a valid nonzero codeword would be $\mathbf{c}_{\Tilde{\mathbf{s}}}=[1,1,0,0]$ (setting $\Tilde{\mathbf{s}}=[0,1]$). 
To generate a valid tag $\Tilde{\mathbf{t}}=[1]$, we have either $\mathbf{k}_1=[1,0,0,0]$ or $\mathbf{k}_1=[0,1,0,0]$.
Thus, the number of appearance of $\Tilde{\mathbf{t}}=[1]$ is $\left|\left\{\mathbf{k}_1\in\mathcal{K}_1: \Tilde{\mathbf{t}}=[1]=\mathbf{c}_{\Tilde{\mathbf{s}},\mathbf{k}_1}\right\}\right|=\left|\{[1,0,0,0],[0,1,0,0]\}\right|=2$.

The following lemma presents an expression for the number of appearances of the tag $\Tilde{\mathbf{t}}$ given certain weights for both the nonzero codeword $\mathbf{c}_{\Tilde{\mathbf{s}}}$ and the tag $\Tilde{\mathbf{t}}$. We denote $\mathrm{wt}(\cdot)$ as the weight of a vector.
\begin{lemma}
\label{lemma:choice1}
     For a given nonzero codeword $\mathbf{c}_{\Tilde{\mathbf{s}}}$ with $\mathrm{wt}(\mathbf{c}_{\Tilde{\mathbf{s}}})=w$, the number of tags $\Tilde{\mathbf{t}}$ with $\mathrm{wt}(\Tilde{\mathbf{t}})=w_t$ is given by
    \begin{equation}   
        \left|\left\{\mathbf{k}_1\in\mathcal{K}_1:\begin{matrix*}[l]\Tilde{\mathbf{t}}=\mathbf{c}_{\Tilde{\mathbf{s}},\mathbf{k}_1},
        \\\mathrm{wt}(\mathbf{c}_{\Tilde{\mathbf{s}}})=w,\mathrm{wt}(\Tilde{\mathbf{t}})=w_t\end{matrix*}
        \right\}\right|={{w \choose w_t}{n-w \choose l-w_t}}.
    \end{equation}
\end{lemma}
\begin{proof}
The proof is by straightforward counting arguments.     
Since the tag $\Tilde{\mathbf{t}}$ is constructed by choosing $l$ entries from the nonzero codeword $\mathbf{c}_{\Tilde{\mathbf{s}}}$, choosing $w_t$ ones from $\mathbf{c}_{\Tilde{\mathbf{s}}}$ has ${ w \choose w_t} $ choices; choosing $l-w_t$ zeros from $\mathbf{c}_{\Tilde{\mathbf{s}}}$ has ${n-w \choose l-w_t}$ choices. Therefore, there are ${ w \choose w_t}{n-w \choose l-w_t}$ choices for constructing a tag $\Tilde{\mathbf{t}}$ with $w_t$ ones and $l-w_t$ zeros.
\end{proof}
There exists a special case in which the nonzero codeword $\mathbf{c}_{\Tilde{\mathbf{s}}}$ with $\mathrm{wt}(\mathbf{c}_{\Tilde{\mathbf{s}}})=w$ can be divided into two sub-blocks where one vector is an all-one vector as $\mathbf{1}_w$, while the other one is an all-zero vector as $\mathbf{0}_{n-w}$, i.e., $\mathbf{c}_{\Tilde{\mathbf{s}}}=[\mathbf{1}_w,\mathbf{0}_{n-w}]$. We analyze the number of appearances of the tags in such special cases in the following lemma.
\begin{lemma}
\label{lemma:choice2}
    Given a nonzero codeword $\mathbf{c}_{\Tilde{\mathbf{s}}}=[\mathbf{1}_w,\mathbf{0}_{n-w}]$ with $\mathrm{wt}(\mathbf{c}_{\Tilde{\mathbf{s}}})=w$, there exists exactly one valid tag $\Tilde{\mathbf{t}}=[\mathbf{1}_{w_t},\mathbf{0}_{l-w_t}]$ with $\mathrm{wt}(\Tilde{\mathbf{t}})=w_t$. Furthermore, for this $\Tilde{\mathbf{t}}$ we have
    \begin{equation}
    \begin{aligned}
        \left|\left\{\mathbf{k}_1\in\mathcal{K}_1:\begin{matrix*}[l]
        \Tilde{\mathbf{t}}=\mathbf{c}_{\Tilde{\mathbf{s}},\mathbf{k}_1}=[\mathbf{1}_w,\mathbf{0}_{n-w}],\\\mathbf{c}_{\Tilde{\mathbf{s}}}=[\mathbf{1}_w,\mathbf{0}_{n-w}],
        \\\mathrm{wt}(\mathbf{c}_{\Tilde{\mathbf{s}}})=w,\mathrm{wt}(\Tilde{\mathbf{t}})=w_t\end{matrix*}
        \right\}\right|={{w \choose w_t}{n-w \choose l-w_t}}.
    \end{aligned}
    \end{equation}
\end{lemma}
\begin{proof}
Note that a tag $\Tilde{\mathbf{t}}$ is constructed by choosing $w_t$ ones and $l-w_t$ zeros from a nonzero codeword $\mathbf{c}_{\Tilde{\mathbf{s}}}$. 
Thus, given the nonzero codeword $\mathbf{c}_{\Tilde{\mathbf{s}}}=[\mathbf{1}_w,\mathbf{0}_{n-w}]$, to construct a tag $\Tilde{\mathbf{t}}$ with $\mathrm{wt}(\Tilde{\mathbf{t}})=w_t$, there is only one valid tag $\Tilde{\mathbf{t}}=[\mathbf{1}_{w_t},\mathbf{0}_{l-w_t}]$ since we should choose $w_t$ ones from $\mathbf{1}_w$ and $l-w_t$ zeros from $\mathbf{0}_{n-w}$, where $\mathbf{1}_w$ and $\mathbf{0}_{n-w}$ are the vectors which constructed the given nonzero codeword $\mathbf{c}_{\Tilde{\mathbf{s}}}=[\mathbf{1}_w,\mathbf{0}_{n-w}]$. 
Then, choosing $w_t$ ones from $\mathbf{1}_w$ has ${w \choose w_t}$ choices; 
choosing $l-w_t$ zeros from $\mathbf{0}_{n-w}$ has ${n-w \choose l-w_t}$ choices. 
Therefore, there are ${ w \choose w_t}{n-w \choose l-w_t}$ choices for constructing a tag $\Tilde{\mathbf{t}}$ of weight $w_t$. 
\end{proof}

Then, we analyze all valid tags given such nonzero codeword in the following lemma.
\begin{lemma}
\label{lemma:cases}
    Given a weight-$w$ nonzero codeword $\mathbf{c}_{\Tilde{\mathbf{s}}}=[\mathbf{1}_w,\mathbf{0}_{n-w}]$, there are four cases of the valid tags $\Tilde{\mathbf{t}}=[\mathbf{1}_{w_t},\mathbf{0}_{l-w_t}]$:
\begin{itemize}
        \item Case 1 ($w\geq l$, $l\geq n-w$): $l\geq w_t\geq l-(n-w)$.
        \item Case 2 ($w< l$, $l\geq n-w$): $w\geq w_t\geq l-(n-w)$.
        \item Case 3 ($w\geq l$, $l< n-w$): $l\geq w_t\geq 0$.
        \item Case 4 ($w< l$, $l< n-w$): $w\geq w_t\geq 0$.
\end{itemize}
\end{lemma}
\begin{proof}
To specify the upper bound on $w_t$, we first consider the relation between $w$ and $l$. When $w\geq l$, the nonzero codeword $\mathbf{c}_{\Tilde{\mathbf{s}}}$ has more ones than the length of the tag $\Tilde{\mathbf{t}}$, which is $l$. Thus, the maximum weight of the tag $\Tilde{\mathbf{t}}$ is $l$ such that $l\geq w_t$; When $w< l$, the nonzero codeword $\mathbf{c}_{\Tilde{\mathbf{s}}}$ has less ones than the length of the tag $\Tilde{\mathbf{t}}$, which is $l$. Thus, the maximum weight of the tag $\Tilde{\mathbf{t}}$ is $w$ such that $w\geq w_t$.

Secondly, to specify the lower bound on $w_t$, we consider the relation between $l$ and $n-w$. When $l\geq n-w$, the nonzero codeword $\mathbf{c}_{\Tilde{\mathbf{s}}}$ has less zeros than the length of the tag $\Tilde{\mathbf{t}}$, which is $l$. 
Thus, the minimum weight of the tag $\Tilde{\mathbf{t}}$ is $l$ such that $w_t\geq l-(n-w)$; When $l< n-w$, the nonzero codeword $\mathbf{c}_{\Tilde{\mathbf{s}}}$ has more zeros than the length of the tag $\Tilde{\mathbf{t}}$, which is $l$. Thus, the minimum weight of the tag $\Tilde{\mathbf{t}}$ is $0$ such that $w_t\geq 0$.

Considering the four cases from the upper bounds and the lower bounds on $w_t$ as stated above completes the proof of the lemma.
\end{proof}

The following lemma states that the maximum $P_\mathrm{t}(\cdot)$ over weight-$w$ nonzero codewords occurs for a nonzero codeword of the type $\mathbf{c}_{\Tilde{\mathbf{s}}}=[\mathbf{1}_w,\mathbf{0}_{n-w}]$, if such a codeword exists.

\begin{lemma}
\label{lemma:choice3}
    Given a nonzero codeword $\mathbf{c}_{\Tilde{\mathbf{s}}}=[\mathbf{1}_w,\mathbf{0}_{n-w}]$ and an arbitrary valid nonzero codeword $\mathbf{c}_{\Tilde{\mathbf{s}}}^\prime$ which $\mathrm{wt}(\mathbf{c}_{\Tilde{\mathbf{s}}})=\mathrm{wt}(\mathbf{c}_{\Tilde{\mathbf{s}}}^\prime)=w$, we have $P_\mathrm{t}(\mathbf{c}_{\Tilde{\mathbf{s}}}=[\mathbf{1}_w,\mathbf{0}_{n-w}])\geq P_\mathrm{t}(\mathbf{c}_{\Tilde{\mathbf{s}}}^\prime)$.
\end{lemma}
\begin{proof}
Based on Lemma~\ref{lemma:choice2}, we have
\begin{equation}
\label{eq:ps1}
\begin{aligned}
    &P_\mathrm{t}(\mathbf{c}_{\Tilde{\mathbf{s}}}=[\mathbf{1}_w,\mathbf{0}_{n-w}])\\
    =&\max_{w_t}\frac{\begin{aligned}
        \left|\left\{\mathbf{k}_1\in\mathcal{K}_1:\begin{matrix*}[l]\Tilde{\mathbf{t}}=\mathbf{c}_{\Tilde{\mathbf{s}},\mathbf{k}_1}=[\mathbf{1}_w,\mathbf{0}_{n-w}],\\\mathbf{c}_{\Tilde{\mathbf{s}}}=[\mathbf{1}_w,\mathbf{0}_{n-w}],
        \\\mathrm{wt}(\mathbf{c}_{\Tilde{\mathbf{s}}})=w,\mathrm{wt}(\Tilde{\mathbf{t}})=w_t\end{matrix*}
        \right\}\right|
    \end{aligned}}{{n \choose l}}\\
    =&\max_{w_t}\frac{{{w \choose w_t}{n-w \choose l-w_t}}}{{n \choose l}}.
\end{aligned}
\end{equation}

By Lemma~\ref{lemma:choice1}, we obtain that \begin{equation}   
        \left|\left\{\mathbf{k}_1\in\mathcal{K}_1:\begin{matrix*}[l]\Tilde{\mathbf{t}}=\mathbf{c}_{\Tilde{\mathbf{s}},\mathbf{k}_1}^\prime,
        \\\mathrm{wt}(\mathbf{c}_{\Tilde{\mathbf{s}}}^\prime)=w,\mathrm{wt}(\Tilde{\mathbf{t}})=w_t\end{matrix*}
        \right\}\right|={{w \choose w_t}{n-w \choose l-w_t}}.
\end{equation}
Let the number of valid tags with $\mathrm{wt}(\mathbf{c}_{\Tilde{\mathbf{s}}}^\prime)=w$ and $\mathrm{wt}(\Tilde{\mathbf{t}})=w_t$ be $T$, i.e.,
\begin{equation}
    \left|\left\{\Tilde{\mathbf{t}}\in\mathcal{T}:\begin{matrix*}[l]\Tilde{\mathbf{t}}=\mathbf{c}_{\Tilde{\mathbf{s}},\mathbf{k}_1}^\prime,
    \\\mathrm{wt}(\mathbf{c}_{\Tilde{\mathbf{s}}}^\prime)=w,\mathrm{wt}(\Tilde{\mathbf{t}})=w_t\end{matrix*}\right\}\right|=T.
\end{equation}
Then we have
\begin{equation}
\begin{aligned}
    \sum_{i=1}^T\left|\left\{\mathbf{k}_1\in\mathcal{K}_1:\begin{matrix*}[l]\Tilde{\mathbf{t}}_i=\mathbf{c}_{\Tilde{\mathbf{s}},\mathbf{k}_1}^\prime,
    \\\mathrm{wt}(\mathbf{c}_{\Tilde{\mathbf{s}}}^\prime)=w,\mathrm{wt}(\Tilde{\mathbf{t}}_i)=w_t\end{matrix*}\right\}\right|
    \\={{w \choose w_t}{n-w \choose l-w_t}}.
\end{aligned}
\end{equation}
Thus, we obtain that 
\begin{equation}
    \left|\left\{\mathbf{k}_1\in\mathcal{K}_1:\begin{matrix*}[l]\Tilde{\mathbf{t}}_i=\mathbf{c}_{\Tilde{\mathbf{s}},\mathbf{k}_1}^\prime,
    \\\mathrm{wt}(\mathbf{c}_{\Tilde{\mathbf{s}}}^\prime)=w,\mathrm{wt}(\Tilde{\mathbf{t}}_i)=w_t\end{matrix*}\right\}\right|\leq{{w \choose w_t}{n-w \choose l-w_t}},
\end{equation}
for $i\in \{1,\dots,T\}$. Therefore, we have
\begin{equation}
\label{eq:ps2}
    P_\mathrm{t}(\mathbf{c}_{\Tilde{\mathbf{s}}}^\prime)
    \leq\max_{w_t}\frac{{{w \choose w_t}{n-w \choose l-w_t}}}{{n \choose l}}.
\end{equation}
Combining \eqref{eq:ps1} and \eqref{eq:ps2}, yields $P_\mathrm{t}(\mathbf{c}_{\Tilde{\mathbf{s}}}=[\mathbf{1}_w,\mathbf{0}_{n-w}])\geq P_\mathrm{t}(\mathbf{c}_{\Tilde{\mathbf{s}}}^\prime)$, for $\mathrm{wt}(\mathbf{c}_{\Tilde{\mathbf{s}}})=\mathrm{wt}(\mathbf{c}_{\Tilde{\mathbf{s}}}^\prime)=w$.
\end{proof}

Furthermore, the following lemma helps us to further reduce the complexity of calculating the probability of success of the substitution attack $P_{\mathrm{S}}$.

\begin{lemma}
\label{lemma:choice4}
    Given $\mathbf{c}_{\Tilde{\mathbf{s}}}=[\mathbf{1}_w,\mathbf{0}_{n-w}]$ and $\mathbf{c}_{\Tilde{\mathbf{s}}}^\prime=[\mathbf{1}_{w^\prime},\mathbf{0}_{n-w^\prime}]$, when $w^\prime=n-w$, we have $P_\mathrm{t}(\mathbf{c}_{\Tilde{\mathbf{s}}}=[\mathbf{1}_w,\mathbf{0}_{n-w}])=P_\mathrm{t}(\mathbf{c}_{\Tilde{\mathbf{s}}}^\prime=[\mathbf{1}_{w^\prime},\mathbf{0}_{n-w^\prime}])$.
\end{lemma}
\begin{proof}
    We have $\mathbf{c}_{\Tilde{\mathbf{s}}}^\prime=(\mathbf{c}_{\Tilde{\mathbf{s}}}+\mathbf{1}_n)\mathbf{R}_n$, where $\mathbf{R}_n$ is a reversal matrix that reverses all $n$ entries in the vector $(\mathbf{c}_{\Tilde{\mathbf{s}}}+\mathbf{1}_n)\mathbf{R}_n$. 
    Accordingly, we have $\Tilde{\mathbf{t}}^\prime=(\Tilde{\mathbf{t}}+\mathbf{1}_l)\mathbf{R}_l$, where $\Tilde{\mathbf{t}}^\prime$ is the corresponding valid tag with $\mathrm{wt}(\Tilde{\mathbf{t}}^\prime)=w_t^\prime$, $\mathbf{R}_l$ is a reversal matrix that reverses all $l$ entries in the vector $\Tilde{\mathbf{t}}+\mathbf{1}_l$, such that $w_t^\prime=l-w_t$. Thus, we have
\begin{equation}
\begin{aligned}
    P_\mathrm{t}(\mathbf{c}_{\Tilde{\mathbf{s}}}^\prime=[\mathbf{1}_{w^\prime},\mathbf{0}_{n-w^\prime}])
    =\max_{w_t^\prime}\frac{{{w^\prime \choose w_t^\prime}{n-w^\prime \choose l-w_t^\prime}}}{{n \choose l}}\\
    \overset{{\scriptstyle\mathrm{(d)}}}{=}\max_{w_t}\frac{{{w \choose w_t}{n-w \choose l-w_t}}}{{n \choose l}}
    =P_\mathrm{t}(\mathbf{c}_{\Tilde{\mathbf{s}}}=[\mathbf{1}_w,\mathbf{0}_{n-w}]),
\end{aligned}
\end{equation}
where $(\mathrm{d})$ is based on $w^\prime=n-w$ and $w_t^\prime=l-w_t$.
\end{proof}

With Lemmas~\ref{lemma:choice3} and \ref{lemma:choice4}, we end the theoretical analysis by presenting the theorem which demonstrates the closed-form expression for $P_{\mathrm{S}}$.

\begin{theorem}
\label{thm:ps}
Given an RM-A-code as defined in Definition 1, we have
    \begin{equation}
    \begin{aligned}
        P_{\mathrm{S}}
        =\max_w\max_{w_t}\frac{{{w \choose w_t}{n-w \choose l-w_t}}}{{n \choose l}},
    \end{aligned}
    \end{equation}
    where the maximization is over all $w$, with $\frac{n}{2^r} \leq w\leq \frac{n}{2}$, for which there exists a codeword of the form $[\mathbf{1}_{w},\mathbf{0}_{n-w}]$.
\end{theorem}
\begin{proof}
    Based on Lemma~\ref{lemma:choice3}, we know that the nonzero codeword $\mathbf{c}_{\Tilde{\mathbf{s}}}$ with the largest $P_\mathrm{t}(\mathbf{c}_{\Tilde{\mathbf{s}}})$ among all valid nonzero codewords $\mathbf{c}_{\Tilde{\mathbf{s}}}^\prime$ with identical weights $\mathrm{wt}(\mathbf{c}_{\Tilde{\mathbf{s}}})=\mathrm{wt}(\mathbf{c}_{\Tilde{\mathbf{s}}}^\prime)=w$ is always the nonzero codeword $\mathbf{c}_{\Tilde{\mathbf{s}}}=[\mathbf{1}_w,\mathbf{0}_{n-w}]$, such that $P_{\mathrm{S}}=\max_{w} P_\mathrm{t}(\mathbf{c}_{\Tilde{\mathbf{s}}}=[\mathbf{1}_w,\mathbf{0}_{n-w}])$.
    Thus, by maximizing $P_\mathrm{t}(\mathbf{c}_{\Tilde{\mathbf{s}}}=[\mathbf{1}_w,\mathbf{0}_{n-w}])$ from choosing the weight $w$ from all valid nonzero codewords with $[\mathbf{1}_w,\mathbf{0}_{n-w}]$'s naturally meets $P_{\mathrm{S}}$. 
    Lemma~\ref{lemma:choice4} indicates that, when $w^\prime=n-w$, we have $P_\mathrm{t}(\mathbf{c}_{\Tilde{\mathbf{s}}}=[\mathbf{1}_w,\mathbf{0}_{n-w}])=P_\mathrm{t}(\mathbf{c}_{\Tilde{\mathbf{s}}}^\prime=[\mathbf{1}_{w^\prime},\mathbf{0}_{n-w^\prime}])$. 
    Then, when $w=w^\prime=\frac{n}{2}$, we have $\mathbf{c}_{\Tilde{\mathbf{s}}}=\mathbf{c}_{\Tilde{\mathbf{s}}}^\prime$. 
    Therefore, the optimization problem only has to consider $w\leq \frac{n}{2}$ as the upper bound. 
    Combining the lower bound which depends on the chosen $\mathrm{RM}(m,r)$ and $M$, as $\frac{n}{2^r} \leq w$, yields $\frac{n}{2^r} \leq w\leq \frac{n}{2}$.
\end{proof}



\begin{remark}
    The proposed projective construction can be, in principle, generalized to all binary linear codes. RM-A-code, in fact, is a special case for constructing such a projective construction for systematic authentication codes. The structure of RM codes allows us to express the range of $w$'s explicitly in the statement of Theorem~\ref{thm:ps}, which reduces the computational overhead significantly compare to the expression for $P_{\mathrm{S}}$ in \eqref{eq:P_S_worst}. For general codes, finding the range of $w$'s for which codewords of the form $[\mathbf{1}_{w},\mathbf{0}_{n-w}]$ exist is difficult and can be exponentially complex. 
\end{remark}

    

To summarize, we would like to highlight the theoretical results of $P_\mathrm{I}$ and $P_\mathrm{S}$ in Theorems~\ref{thm:pi} and \ref{thm:ps}, respectively, for the proposed projective construction, RM-A-codes, as follows:
\begin{equation}
    P_\mathrm{I}=\frac{1}{2^l},
\end{equation}
and
\begin{equation}
    \begin{aligned}
        P_{\mathrm{S}}
        =\max_w\max_{w_t}\frac{{{w \choose w_t}{n-w \choose l-w_t}}}{{n \choose l}},
    \end{aligned}
\end{equation}
where $\frac{n}{2^r} \leq w\leq \frac{n}{2}$.

\subsection{Numerical Analysis}

Next, we provide numerical results for the proposed RM-A-codes with different blocklengths. The results are shown in TABLE~\ref{tab:num_rates}. In this setting, the source length is set as $M=4$, tag length as $l=3$, and the order of the RM code is $r=1$, together with $m=\{4,5,6,7,8\}$. In TABLE~\ref{tab:num_rates}, it can be observed that $P_{\mathrm{S}}$ decreases while the blocklength increases. Furthermore, since $l$ is fixed as $l=3$, we have $P_{\mathrm{I}}$ as a constant $P_{\mathrm{I}}=\frac{1}{2^l}=0.125$, which meets the lower bound. We leave a more thorough numerical analysis for our future works.
\begin{table}[h]  
\normalsize
\centering
\begin{tabular}{|c||c|c|c|c|c|}
 \hline
$m$&$4$&$5$& $6$ & $7$ & $8$\\
 \hline
 Size& $(16,5)$ & $(32,6)$ & $(64,7)$ & $(128,8)$ & $(256,9)$\\
 \hline\hline
 $P_{\mathrm{I}}$& $0.125$ & $0.125$ & $0.125$ & $0.125$ & $0.125$\\ 
 \hline
 $P_{\mathrm{S}}$& $0.4$ & $0.3817$ & $0.3810$ & $0.3780$ & $0.3765$\\
 \hline
\end{tabular}
\caption{$M=4$, $l=3$, and $r=1$, varies $m$}
  \label{tab:num_rates}
\end{table}




\section{Conclusion}
\label{sec:con}
In this paper, we proposed a projective construction of systematic authentications based on binary linear codes, and studied a particular case based on RM codes, referred to as the RM-A-codes. The theoretical results are provided for the probabilities of deception. Furthermore, we have discussed explicit connections between the probability of success for the substitution attack and the RM code structure, which captures certain properties in the structures of error-correcting codes that are not very well understood. 
A potential direction for future work is to extend the projective construction for systematic authentications to more general classes of binary linear codes.

\appendix

\subsection{Linearity of RM-A-codes}
\label{app:linearity}
\begin{lemma}[linearity]
\label{lemma:linear}
    Given $\mathbf{s},\mathbf{s}^\prime\in\mathcal{S}=\{0,1\}^M$ and a generator matrix $\mathbf{G}$ from $\mathrm{RM}(m,r)$, we have $\mathbf{c}_{\mathbf{s}}+\mathbf{c}_{\mathbf{s}^\prime}=\mathbf{c}_{\mathbf{s}+\mathbf{s}^\prime}$.
\end{lemma}
\begin{proof}
   Recall that $\mathbf{c}_\mathbf{s}=[\mathbf{0}_{\scriptscriptstyle \sum_{i=0}^r{m \choose i}-M-1},\mathbf{s},0]\mathbf{G}$ and $\mathbf{c}_{\mathbf{s}^\prime}=[\mathbf{0}_{\scriptscriptstyle \sum_{i=0}^r{m \choose i}-M-1},\mathbf{s}^\prime,0]\mathbf{G}$. Thus, we have
   \begin{equation*}
   \begin{aligned}
\mathbf{c}_\mathbf{s}+\mathbf{c}_{\mathbf{s}^\prime}
       =&[\mathbf{0}_{\scriptscriptstyle \sum_{i=0}^r{m \choose i}-M-1},\mathbf{s},0]\mathbf{G}+[\mathbf{0}_{\scriptscriptstyle \sum_{i=0}^r{m \choose i}-M-1},\mathbf{s}^\prime,0]\mathbf{G}\\
       =&[\mathbf{0}_{\scriptscriptstyle \sum_{i=0}^r{m \choose i}-M-1},\mathbf{s}+\mathbf{s}^\prime,0]\mathbf{G}=\mathbf{c}_{\mathbf{s}+\mathbf{s}^\prime}.
    \end{aligned}
    \end{equation*}
    The proof simply follows from the linearity of the code and also holds for any other linear codes.
\end{proof}

\bibliographystyle{IEEEtran}
\bibliography{Paper} 
\end{document}